\pdfoutput=1

\newcommand*{\ADVERSARY}{\textsc{adversary}}
\newcommand*{\NN}{\mathbb{N}}
\newcommand*{\PAD}[2][\;]{#1#2#1}
\newcommand*{\SET}[1]{\{\,#1\,\}}
\newcommand*{\SETC}[2]{\SET{#1\mid#2}}
\newcommand*{\TREES}[1]{T_{#1}}

\makeatletter
\newenvironment{mcqm@figure}[2]{%
    \@tempcnta=932340 % 5 mm in sp
    \setlength{\unitlength}{#1\@tempcnta}%
    \sbox{\z@}{\includegraphics[scale=#1]{#2}}%
    \settodepth{\@tempcnta}{\usebox{\z@}}%
    \settoheight{\@tempcntb}{\usebox{\z@}}%
    \advance\@tempcntb\@tempcnta
    \settowidth{\@tempcnta}{\usebox{\z@}}%
    \divide\@tempcnta by\unitlength
    \divide\@tempcntb by\unitlength
    \begin{picture}(\@tempcnta,\@tempcntb)
        \put(0,0){\makebox(0,0)[bl]{\usebox{\z@}}}}{%
    \end{picture}}
\newenvironment{FIGURE}[1]{%
    \begin{mcqm@figure}{0.6}{#1}}{\end{mcqm@figure}}
\newcommand*{\LABEL}[4][2]{%
    \put(#2,#3){\makebox(#1,2){\footnotesize\vphantom{$()$}#4}}}
\makeatother

\documentclass[11pt]{article}

\usepackage[utf8]{inputenc}
\usepackage[british]{babel}

\usepackage{amsfonts}
\usepackage{amsmath}
\usepackage{graphicx}
\usepackage[round]{natbib}
\usepackage[amsmath,standard,thmmarks]{ntheorem}
\usepackage{subfig}

\begin{document}

\title{%
    Ogden's Lemma for Regular Tree Languages}

\author{%
    Marco Kuhlmann\thanks{I wish to thank Mathias M\"ohl and an anonymous reviewer for pointing out errors, and for comments that helped to improve the quality of the presentation. The work reported in this paper was partially funded by the German Research Foundation.}\\
    Dept.\ of Linguistics and Philology\\
    Uppsala University, Sweden}

\maketitle

\section{Introduction}%
\label{sec:Introduction}

Pumping lemmata are elementary tools for the analysis of formal languages. While they usually cannot be made strong enough to fully capture a class of languages, it is generally desirable to have as strong pumping lemmata as possible. However, this is counterbalanced by the experience that strong pumping lemmata may be hard to prove, or, worse, hard to use---this experience has been made, for example, in the study of the output languages of tree transducers, where other proof techniques, so-called \emph{bridge theorems}, make better tools \citep{engelfriet2002output}. The purpose of this squib is to strengthen the standard pumping lemma for the class of \emph{regular tree languages} \citep{gecseg1997tree}, without sacrificing its usability, in the same way as Ogden strengthened the pumping lemma for context-free string languages \citep{ogden1968helpful}.

The paper is structured as follows. Section~\ref{sec:Preliminaries} introduces our notation. Section~\ref{sec:Motivation} presents the main lemma and motivates it using a small, formal example. Finally, Section~\ref{sec:Proofs} contains the proof of the main lemma.

\section{Preliminaries}%
\label{sec:Preliminaries}

We assume the reader to be familiar with the standard concepts from the theory of tree languages. The notation that we use in this paper is mostly identical to the one used in the survey by \citet{gecseg1997tree}; the major difference is our way of denoting substitution into contexts.

We write~$\NN$ for the set of non-negative natural numbers, and $[n]$ as an abbreviation for the set $\SETC{i\in\NN}{1\leq i\leq n}$. Given a set~$A$, we write $|A|$ for the cardinality of~$A$, and $A^*$ for the set of all strings over~$A$.

Let~$\Sigma$ be a ranked alphabet. For a tree $t\in\TREES{\Sigma}$, we write $|t|$ to denote the \emph{size of~$t$}, defined as the number of nodes of~$t$. A \emph{path in~$t$} is a sequence of nodes of~$t$ in which each node but the first one is a child of the node preceding it. Let~${\circ}$ be a symbol with rank zero that does not occur in~$\Sigma$. Recall that a \emph{context over~$\Sigma$} is a tree~$c$ over $\Sigma\cup\SET{{\circ}}$ in which~${\circ}$ occurs exactly once. We call the (leaf) node at which the symbol~${\circ}$ occurs the \emph{hole} of the context. We write $|c|$ to denote the \emph{size of the context~$c$}, defined as the number of non-hole nodes of~$c$. Finally, given a context~$c$ and a tree~$t$, we write $c\cdot t$ for the tree obtained by substituting~$t$ into~$c$ at its hole. Note that Gécseg and Steinby denote this tree by $t\cdot c$ or $c(t)$.

A subset $L\subseteq\TREES{\Sigma}$ is a \emph{tree language over~$\Sigma$}. A tree language is \emph{regular}, if there is a finite-state tree automaton that accepts~$L$.

\section{Motivation}%
\label{sec:Motivation}

To motivate the need for a strong pumping lemma for regular tree languages, we start with a look at the standard one \citep{gecseg1997tree}\footnote{The lemma given here is in fact slightly stronger than the one given by \citet{gecseg1997tree} (Proposition 5.2), and makes pumpability dependent on the size of a tree, rather than on its height.}:

\begin{lemma}\label{lem:PumpingLemma}%
    For every regular tree language $L\subseteq\TREES{\Sigma}$, there is a number $p\geq 1$ such that any tree $t\in L$ of size at least~$p$ can be written as $t=c'\cdot c\cdot t'$ in such a way that $|c|\geq 1$, $|c\cdot t'|\leq p$, and $c'\cdot c^n\cdot t'\in L$, for every $n\in\NN$.
\end{lemma}

\noindent Just as the pumping lemma for context-free string languages, Lemma~\ref{lem:PumpingLemma} is most often used in its contrapositive formulation, which specifies a strategy for proofs that a language $L\subseteq\TREES{\Sigma}$ is \emph{not} regular: show that, for all $p\geq 1$, there exists a tree $t\in L$ of size at least~$p$ such that for any decomposition $c'\cdot c\cdot t'$ of~$t$ in which $|c|\geq 1$ and $|c\cdot t'|\leq p$, there is a number $n\in\NN$ such that $c'\cdot c^n\cdot t'\notin L$. It is helpful to think of a proof according to this strategy as a game against an imagined \ADVERSARY, where our objective is to prove that~$L$ is non-regular, and \ADVERSARY's objective is to foil this proof. The game consists of four alternating turns: In the first turn, \ADVERSARY\ must choose a number $p\geq 1$. In the second turn, we must respond to this choice by providing a tree $t\in L$ of size at least~$p$. In the third turn, \ADVERSARY\ must choose a decomposition of~$t$ into fragments $c'\cdot c\cdot t'$ such that $|c|\geq 1$ and $|c\cdot t'|\leq p$. In the fourth and final turn, we must provide a number $n\in\NN$ such that $c'\cdot c^n\cdot t'\notin L$. If we are able to do so, we win the game; otherwise, \ADVERSARY\ wins. We can prove that~$L$ is non-regular, if we have a winning strategy for the game.

\begin{figure}
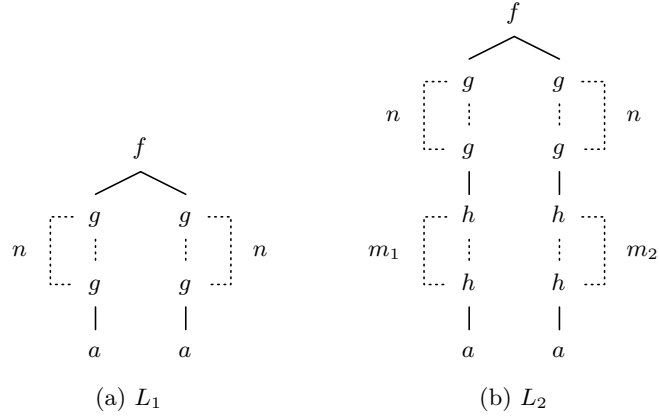

	\bigskip\centering
	\subfloat[$L_1$]{%
		\label{fig:PumpingLemmaExample1}%
		\begin{FIGURE}{PumpingLemmaExample1}
			\LABEL{5}{09}{$f$}
			\LABEL{3}{06}{$g$}
			\LABEL{3}{03}{$g$}
			\LABEL{3}{00}{$a$}
			\LABEL{7}{06}{$g$}
			\LABEL{7}{03}{$g$}
			\LABEL{7}{00}{$a$}
			\LABEL{00}{04.5}{\llap{$n$}}
			\LABEL{10}{04.5}{\rlap{$n$}}
		\end{FIGURE}}
	\qquad\qquad
	\subfloat[$L_2$]{%
		\label{fig:PumpingLemmaExample2}%
		\begin{FIGURE}{PumpingLemmaExample2}
			\LABEL{5}{15}{$f$}
			\LABEL{3}{12}{$g$}
			\LABEL{3}{09}{$g$}
			\LABEL{3}{06}{$h$}
			\LABEL{3}{03}{$h$}
			\LABEL{3}{00}{$a$}
			\LABEL{7}{12}{$g$}
			\LABEL{7}{09}{$g$}
			\LABEL{7}{06}{$h$}
			\LABEL{7}{03}{$h$}
			\LABEL{7}{00}{$a$}
			\LABEL{00}{10.5}{\llap{$n$}}
			\LABEL{10}{10.5}{\rlap{$n$}}
			\LABEL{00}{04.5}{\llap{$m_1$}}
			\LABEL{10}{04.5}{\rlap{$m_2$}}
		\end{FIGURE}}
	\caption{Two tree languages that are not regular}
\end{figure}

Consider the language $L_1 = \SETC{f(g^n\cdot a,g^n\cdot a)}{n\geq 1}$, shown schematically in Figure~\ref{fig:PumpingLemmaExample1}. Using Lemma~\ref{lem:PumpingLemma}, it is easy to show that this language is non-regular: we always win by presenting \ADVERSARY\ with the tree $t = f(g^p\cdot a,g^p\cdot a)$. To see this, notice that in whatever way \ADVERSARY\ decomposes~$t$ into fragments $c'\cdot c\cdot t'$ such that $|c|\geq 1$, the pumped tree $c'\cdot c^2\cdot t'$ does not belong to~$L_1$. In particular, if~$c$ is rooted at a node that is labelled with~$g$, then the pumped tree violates the constraint that the two branches have the same length.

Unfortunately, Lemma~\ref{lem:PumpingLemma} sometimes is too blunt a tool to show the non-regularity of a tree language. Consider the language
\begin{displaymath}
    L_2 \PAD{=}
    \SETC{f(g^n\cdot h^{m_1}\cdot a,g^n\cdot h^{m_2}\cdot a)}{n,m_1,m_2\geq 1}
\end{displaymath}
(see Figure~\ref{fig:PumpingLemmaExample2}). It is not unreasonable to believe that~$L_2$, like~$L_1$, is non-regular, but it is impossible to prove this using Lemma~\ref{lem:PumpingLemma}. To see this, notice that \ADVERSARY\ has a winning strategy for $p\geq 2$: for every tree $t\in L_2$ that we can provide in the second turn of the game, \ADVERSARY\ can choose any decomposition $c'\cdot c\cdot t'$ in which $c=h({\circ})$ and $t'=a$. In this case, $|c|\geq 1$, $|c\cdot t'|\leq p$, and both deleting and pumping~$c$ yield only valid trees in~$L_2$. Intuitively, we would like to force \ADVERSARY\ to choose a decomposition that contains a $g$-labelled node, thus transferring our winning strategy for~$L_1$---but this is not warranted by Lemma~\ref{lem:PumpingLemma}, which merely asserts that a pumpable context does exist \emph{somewhere} in the tree, but does not allow us to delimit the exact region. The pumping lemma that we prove in this paper makes a stronger assertion:

\begin{lemma}\label{lem:Ogden}%
    For every regular tree language $L\subseteq\TREES{\Sigma}$, there is a number $p\geq 1$ such that every tree $t\in L$ in which at least~$p$ nodes are marked as distinguished can be written as $t = c'\cdot c\cdot t'$ such that at least one node in~$c$ is marked, at most~$p$ nodes in $c\cdot t'$ are marked, and $c'\cdot c^n\cdot t'\in L$, for all $n\in\NN$.
\end{lemma}

\noindent Note that, in the special case where all nodes are marked, Lemma~\ref{lem:Ogden} reduces to Lemma~\ref{lem:PumpingLemma}.

Lemma~\ref{lem:Ogden} can be seen as the natural correspondent of Ogden's Lemma for context-free string languages \citep{ogden1968helpful}. Its contrapositive corresponds to the following modified game for tree languages~$L$: In the first turn, \ADVERSARY\ has to choose a number $p\geq 1$. In the second turn, we have to choose a tree $t\in L$ and mark at least~$p$ nodes in~$t$. In the third turn, \ADVERSARY\ has to choose a decomposition $c'\cdot c\cdot t'$ of~$t$ \emph{in such a way that at least one node in~$c$ and at most~$p$ nodes in $c\cdot t'$ are marked}. In the fourth and final turn, we have to choose a number $n\in\NN$ such that $c'\cdot c^n\cdot t'\notin L$. In this modified game, we can implement our idea from above to prove that the language~$L_2$ is non-regular: we can always win the game by presenting \ADVERSARY\ with the tree $t = f(g^p\cdot h(a),g^p\cdot h(a))$ and marking all nodes that are labelled with~$g$ as distinguished. Then, in whatever way \ADVERSARY\ decomposes~$t$ into segments $c'\cdot c\cdot t'$, the context~$c$ contains at least one node labelled with~$g$, and the tree $c'\cdot c^2\cdot t'$ does not belong to~$L_2$.

\section{Proof}%
\label{sec:Proofs}

Our proof of Lemma~\ref{lem:Ogden} builds on the following technical lemma:

\begin{lemma}\label{lem:Technical}%
    Let~$\Sigma$ be a ranked alphabet. For every tree language $L\subseteq\TREES{\Sigma}$ and every $k\geq 1$, there exists a number $p\geq 1$ such that every tree $t\in L$ in which at least~$p$ nodes have been marked as distinguished can be written as $t = c'\cdot c_1\cdots c_k\cdot t'$ in such a way that for each $i\in [k]$, the context~$c_i$ contains at least one marked node, and the tree $c_1\cdots c_k\cdot t'$ contains at most~$p$ marked nodes.
\end{lemma}

\begin{proof}
    Let~$m$ be the maximal rank of any symbol in~$\Sigma$. Note that if~$m$ is zero, then each tree over~$\Sigma$ has size one, and the lemma trivially holds with $p=2$. For the remainder of the proof, assume that $m\geq 1$. Put $g_\Sigma(n)=\sum_{i=0}^n m^i$, and note that $g_\Sigma(n)<g_\Sigma(n+1)$, for all $n\in\NN$. We will show that we can choose $p=g_\Sigma(k)$.
    
    Let $t\in L$ be a tree in which at least one node has been marked as distinguished. We call a node~$u$ of~$t$ \emph{interesting}, if it either is marked, or has at least two children from which there is a path to an interesting node. It is easy to see from this definition that from every interesting node, there is a path to a marked node. Let $d(u)$ denote the number of interesting nodes on the path from the root node of~$t$ to~$u$, excluding~$u$ itself. We make two observations about the function $d(u)$:
    
    First, there is exactly one interesting node~$u$ with $d(u)=0$. To see that there is at most one such node, let~$u_1$ and~$u_2$ be distinct interesting nodes with $d(u_1)=d(u_2)=n$; then the least common ancestor~$u$ of~$u_1$ and~$u_2$ is an interesting node with $d(u)=n-1$. To see that there is at least one such node, recall that every marked node is interesting.
    
    For the second observation, let~$u$ be an interesting node with $d(u)=n$. The number of interesting descendants~$v$ of~$u$ with $d(v)=n+1$ is at most~$m$. To see this, notice that each path from~$u$ to~$v$ starts with~$u$, continues with some child~$u'$ of~$u$, and then visits only non-interesting nodes~$w$ until reaching~$v$. From each of these non-interesting nodes~$w$, there is at most one path that leads to~$v$. Therefore, the path from~$u$ to~$v$ is uniquely determined except for the choice of the child~$u'$, which is a choice among at most~$m$ alternatives.
    
    Taken together, these observations imply that the number of interesting nodes~$u$ with $d(u)\leq k-1$ is bounded by the value $g_\Sigma(k-1)$.
    
    Now, let~$t$ be a tree in which at least $g_\Sigma(k)$ nodes have been marked as distinguished. Then there is at least one interesting node~$u$ with $d(u)=k$, and hence, at least one path that visits at least $k+1$ interesting nodes. Choose any path that visits the maximal number of interesting nodes, and let~$\vec{u}$ be a suffix of that path that visits exactly $k+1$ interesting nodes, call them $v_1,\dots,v_{k+1}$. We use~$\vec{u}$ to identify a decomposition $c_1\cdots c_k\cdot t'$ of~$t$ as follows: for each $i\in [k]$, choose~$v_i$ as the root node of~$c_i$, choose $v_{i+1}$ as the hole of $c_i$, and choose $v_{k+1}$ as the root node of~$t'$. This decomposition satisfies the required properties: To see that the tree $c_1\cdots c_k\cdot t'$ contains at most~$p$ marked nodes, notice that, by the choice of~$\vec{u}$, no path in~$t$ that starts at~$v_1$ contains more than $k+1$ interesting nodes, and hence the total number of interesting nodes in the subtree rooted at~$v_1$ is bounded by $g_\Sigma(k)=p$. To see that every context~$c_i$, $i\in [k]$, contains at least one marked node, let~$v$ be one of the interesting nodes in~$c_i$, and assume that~$v$ is not itself marked. Then~$v$ has at least two children from which there is a path to an interesting, and, ultimatively, to a marked node. At most one of these paths visits~$v_{i+1}$; the marked node at the end of the other path is a node of~$c_i$.
\end{proof}

\noindent With Lemma~\ref{lem:Technical} at hand, the proof of Lemma~\ref{lem:Ogden} is straightforward, and essentially identical to the proof given for the standard pumping lemma \citep{gecseg1997tree}:

\begin{proof}[of Lemma~\ref{lem:Ogden}]%
    Let $L\subseteq\TREES{\Sigma}$ be a regular tree language, and let~$M$ be a tree automaton with state set~$Q$ that recognizes~$L$. We will apply Lemma~\ref{lem:Technical} with $k=|Q|$. Let $t\in L$ be a tree in which at least~$p$ nodes are marked as distinguished, where~$p$ is the number from Lemma~\ref{lem:Technical}. Then~$t$ can be written as $c'\cdot c_1\cdots c_k\cdot t'$ such that for each index $i\in [k]$, the context~$c_i$ contains at least one marked node, and the tree $c_1\cdots c_k\cdot t'$ contains at most~$p$ marked nodes. Note that each context~$c_i$, $i\in [k]$, is necessarily non-empty. Since~$M$ has only~$k$ states, it must arrive in the same state at the root nodes of at least two contexts $c_i$, $i\in [k]$, or at the root node of some context $c_i$, $i\in [k]$, and the root node of~$t'$. A decomposition of~$t$ of the required kind is then obtained by cutting~$t$ at these two nodes.
\end{proof}

\noindent Note that by choosing $k=m\cdot |Q|$ in this proof, where $m\geq 1$, it is easy to generalize Lemma~\ref{lem:Ogden} as follows:
\begin{lemma}\label{lem:OgdenStrong}%
    For every regular tree language $L\subseteq\TREES{\Sigma}$ and every $m\geq 1$, there is a number $p\geq 1$ such that every tree $t\in L$ in which at least~$p$ nodes are marked as distinguished can be written as $t = c'\cdot c_1\cdots c_m\cdot t'$ such that for each $i\in [m]$, at least one node in~$c_i$ is marked, at most~$p$ nodes in $c_1\cdots c_m\cdot t'$ are marked, and $c'\cdot c_1^n\cdots c_m^n\cdot t'\in L$, for all $n\in\NN$.
\end{lemma}

\bibliographystyle{plainnat}
\bibliography{mcqm}

\end{document}